\documentclass[journal,11pt,draftcls,onecolumn]{IEEEtran}
\usepackage{amsfonts}
\usepackage{amsmath}
\usepackage{graphicx}
\usepackage{url}
\usepackage{eepic, epsfig, amsmath, amssymb, latexsym, setspace, subfig}
\usepackage{rotating}
\usepackage{amsfonts}
\usepackage{amsmath}
\usepackage{graphicx}
\usepackage{nicefrac}
\usepackage{lipsum,multicol}

\newtheorem{theorem}{Theorem}[section]
\newtheorem{definition}[theorem]{Definition}
\newtheorem{lemma}[theorem]{Lemma}

\newtheorem{corollary}[theorem]{Corollary}

\numberwithin{equation}{section}

\newcommand{\qed}{\rule{7pt}{7pt}}
\newenvironment{proof}{\noindent{\bf Proof}\hspace*{1em}}{\hfill\qed\vspace{0.125in}}

\newcommand{\x}{\mathbf{x}}
\newcommand{\y}{\mathbf{y}}
\newcommand{\z}{\mathbf{z}}
\newcommand{\w}{\mathbf{w}}
\newcommand{\e}{\mathbf{e}}

\def\bfe{{\mathbf e}}

\def\bfh{{\mathbf h}}

\def\bfv{{\mathbf v}}
\def\bfw{{\mathbf w}}
\def\bfx{{\mathbf x}}
\def\bfy{{\mathbf y}}
\def\bfz{{\mathbf z}}

\begin{document}
%
\title{Sparse Error Correction from Nonlinear Measurements with Applications in Bad Data Detection for Power Networks}

\author{Weiyu Xu, Meng Wang, Jianfeng Cai and Ao Tang \\
\thanks{Part of this paper was presented in the 50th IEEE Conference on Decision and Control 2011  \cite{CDC2011}. Weiyu Xu and Jianfeng Cai are with the University of Iowa; Meng Wang is with the Rensselaer Polytechnic Institute; and Ao Tang is with Cornell University.}}


%

\maketitle

\begin{abstract}
In this paper, we consider the problem of sparse recovery from nonlinear measurements, which has applications in state estimation and bad data detection for power networks. An iterative mixed $\ell_1$ and $\ell_2$ convex program is used to estimate the true state by locally linearizing the nonlinear measurements. When the measurements are linear, through using the almost Euclidean property for a linear subspace, we derive a new performance bound for the state estimation error under sparse bad data and additive observation noise. As a byproduct, in this paper we provide sharp bounds on the almost Euclidean property of a linear subspace, using the ``escape-through-the-mesh'' theorem from geometric functional analysis.  When the measurements are nonlinear, we give conditions under which the solution of the iterative algorithm converges to the true state even though the locally linearized measurements may not be the actual nonlinear measurements. We numerically evaluate our iterative convex programming approach to perform bad data detections in nonlinear electrical power networks problems. We are able to use semidefinite programming to verify the conditions for convergence of the proposed iterative sparse recovery algorithms from nonlinear measurements.
\end{abstract}

\IEEEpeerreviewmaketitle

\section{Introduction}
In this paper, inspired by state estimation for nonlinear electrical power networks under bad data and additive noise, we study the problem of sparse recovery from nonlinear measurements.
The static state of an electric power network can be described by the vector of bus voltage magnitudes and angles. In smart grid power networks, indirect nonlinear measurement results of these quantities are sent to the central control center, where the state estimation of electric power network was performed. On the one hand, these measurement results contain common small additive noises due to the accuracy of meters and equipments. On the other hand, more severely, these results can contain gross errors due to faulty sensors, meters and system malfunctions. In addition, erroneous communications and adversarial compromises of the meters can also introduce gross errors to the measurement quantities received by the central control center. In these scenarios, the observed measurements contain abnormally large measurement errors, called bad data, in addition to the usual additive observation noises. So the state estimation in power networks needs to detect, identify, and eliminate these large measurement errors \cite{BC,FB,MG}. While there are a series of works for dealing with outliers in \emph{linear} measurements \cite{CandesErrorCorrection, CT1}, the measurements for state estimation in power networks are \emph{nonlinear} functions of the states. This motivates us to study the general problem of state estimation from nonlinear measurements and bad data.

Suppose that we make $n$ measurements to estimate the state $\x$ described by an $m$-dimensional ($m<n$) real-numbered  vector, then these measurements can be written as an $n$-dimensional vector $\y$, which is related to the state vector through the measurement equation
\begin{equation}
\y=h(\x)+\bfv+\e,
\label{eq: powermodel}
\end{equation}
where $h(\x)$ is a set of $n$ general functions, which may be linear or a nonlinear, and $\bfv$ is the vector of additive measurement noise, and $\e$ is the vector of bad data imposed on the measurements. In this paper, we assume that $\bfv$ is an $n$-dimensional vector with i.i.d. zero mean Gaussian elements of variance $\sigma^2$. We also assume that $\e$ is a vector with at most $k$ nonzero entries, and the nonzero entries can take arbitrary real-numbered values. The sparsity $k$ of gross errors reflects the nature of bad data because generally only a few faulty sensing results are present or an adversary party may control only a few malicious meters. It is natural that a small number of sensors and meters are faulty at a certain moment; an adversary party may be only able to alter the results of a limited number of meters under his control; and communication errors of meter results are often rare.

When there are no bad data present, it is well known that the Least Square (LS) method can be used to suppress the effect of observation noise on state estimations. For this problem, we need the nonlinear LS method, where we try to find a vector $\x$ minimizing
\begin{equation}
\|\y-h(\x)\|_{2}.
\label{eq:ls}
\end{equation}
However, the LS method generally only works well when there are no bad data $\e$ corrupting the observation $\y$. If the magnitudes of bad data are large, the estimation result can be very far from the true state. So further techniques to eliminate abnormal measurements is needed when there are bad data present in the measurement results.

Bad data detection in power networks can be viewed as a sparse error detection problem, which shares similar mathematical structures as sparse recovery problems in compressive sensing \cite{ CandesErrorCorrection,CT1}. However, this problem in power networks is very different from linear sparse error detection problem \cite{CT1}. In fact, $h(\x)$ in (\ref{eq: powermodel}) is a nonlinear mapping instead of a linear mapping as in \cite{CandesErrorCorrection}.  It is the goal of this paper to provide a sparse recovery algorithm and performance analysis for sparse recovery from nonlinear measurements with applications in bad data detection for electrical power networks.

Toward this end, we first consider the simplified problem when $h(\x)$ is linear, which serves as a basis for solving and analyzing sparse recovery problems with nonlinear measurements. For this sparse recovery problem with linear measurements, a mixed least $\ell_1$ norm and $\ell_2$ norm convex program is used to simultaneously detect bad data and subtract additive noise from the observations. In our theoretical analysis of the decoding performance, we assume $h(\x)$ is a linear transformation $H\x$, where $H$ is an $n \times m$ matrix with i.i.d. standard zero mean Gaussian entries. Through using the almost Euclidean property for the linear subspace generated by $H$, we derive a new performance bound for the state estimation error under sparse bad data and additive observation noise.  In our analysis,  using the ``escape-through-a-mesh'' theorem from geometric functional analysis \cite{Gordon}, we are able to significantly improve on the bounds for the almost Euclidean property of a linear subspace, which may be interesting in a more general mathematical setting. Compared with earlier analysis on the same optimization problem in \cite{CandesErrorCorrection}, we are able to give explicit bounds on the error performance, which is generally sharper than the result in \cite{CandesErrorCorrection} in terms of recoverable sparsity.

 We then consider the nonlinear measurement setting. Generalizing the algorithm and results for linear measurements, we propose an iterative convex programming approach to perform joint noise reduction and bad data detection from nonlinear measurements. We establish conditions under which the iterative algorithm converges to the true state in the presence of bad data even when the measurements are nonlinear. We are also able to verify explicitly when the conditions hold through a semidefinite programming formulation. Our iterative convex programming based algorithm is shown to work well in this nonlinear setting by numerical examples. Compared with \cite{KV82}, which proposed to apply $\ell_1$ minimization in bad data detection in power networks, our approach offers a better decoding error performance when both bad data and additive observation noises are present. \cite{TL10a}\cite{TL10b} considered state estimations under malicious data attacks, and formulated state estimation under malicious attacks as a hypothesis testing problem by assuming a prior probability distribution on the state $\x$. In contrast, our approach does not rely on any prior information on the signal $\x$ itself, and the performance bounds hold for an arbitrary state $\x$. Compressive sensing with nonlinear measurements were studied in \cite{Blumensath} by extending the restricted isometry condition. Our sparse recovery problem is different from the compressive sensing problem considered in \cite{Blumensath} since our measurements are overcomplete and are designed to perform sparse error corrections instead of compressive sensing. Our analysis also does not rely on extensions of the restricted isometry condition.

The rest of this paper is organized as follows. In Section \ref{sec:condition}, we study joint bad data detection and denoising for linear measurements,  and derive the  performance bound on the decoding error based on the almost Euclidean property of linear subspaces. In Section \ref{sec:boundingEuclidean}, a sharp bound on the almost Euclidean property is given through the ``escape-through-mesh'' theorem.  In Section \ref{sec:evaluating}, we present explicitly computed bounds on the estimation error for linear measurements. In Section \ref{sec:nonlinear}, we propose our iterative convex programming algorithm to perform sparse recovery from nonlinear measurements and give theoretical analysis on the performance guarantee of the iterative algorithm. In Section \ref{sec:numerical}, we  present simulation results of our iterative algorithm to show its performance in power networks. 

\section{Bad Data Detection for Linear Systems}
\label{sec:condition}
In this section, we introduce a convex programming formulation to do bad data detection in linear systems, and characterize its decoding error performance. In a linear system, the corresponding $n \times 1$ observation vector in (\ref{eq: powermodel}) is $\y=H\x+\e+\bfv$, where $\x$ is an $m \times 1$ signal vector ($m<n$),  $H$ is an $n \times m$ matrix,  $\e$ is a sparse error vector with at most $k$ nonzero elements, and $\bfv$ is a noise vector with $\|\bfv\|_2 \leq \epsilon$. In what follows, we denote the part of any vector $\bfw$ over any index set $K$ as $\bfw_{K}$.

We solve the following optimization problem involving optimization variables $\x$ and $\z$,  and we then estimate the state $\x$ to be $\hat{\x}$, which is the optimizer value for $\x$.
\begin{eqnarray}
\min_{\x,\z}  &&\|\y-H\x-\z\|_{1},\nonumber \\
{\text{subject to}}&&\|\z\|_2 \leq \epsilon.
\label{eq:bus}
\end{eqnarray}
This optimization problem was proposed in a slightly different form in \cite{CandesErrorCorrection} by restricting $\z$ in the null space of $H^{T}$. We are now ready to give a theorem which bounds the decoding error performance of (\ref{eq:bus}), using the almost Euclidean property \cite{GarnaevGluskin,Kasin}.

\begin{definition} [Almost Euclidean Property]
A subspace in $R^n$ satisfies the \emph{almost Euclidean} property for a constant $\alpha \leq 1$, if
\begin{equation*}
\alpha \sqrt{n} \|\bfw\|_2 \leq \|\bfw\|_{1}
\end{equation*}
holds true for every $\bfw$ in the subspace.

\end{definition}

\begin{theorem}
Let $\y$, $H$, $\x$, $\e$ and $\bfv$ be specified as above. Suppose that the minimum nonzero singular value of $H$ is $\sigma_{\text{min}}$. Let $C$ be a real number larger than $1$, and suppose that every vector $\bfw$ in range of the matrix $H$ satisfies $C\|\bfw_K\|_{1} \leq \|\bfw_{\overline{K}}\|_{1}$ for any subset $K \subseteq \{1,2,...,n\}$ with cardinality $|K|\leq k$, where $k$ is an integer, and $\overline{K}=\{1,2,...,n\}\setminus K$. We also assume the subspace generated by $H$ satisfies the \emph{almost Euclidean} property for a constant $\alpha \leq 1$.

Then the solution $\hat{\x}$ to (\ref{eq:bus}) satisfies
\begin{equation}
\|\x-\hat{\x}\|_{2} \leq \frac{2(C+1)}{\sigma_{\text{min}} \alpha (C-1)}\epsilon.
\label{eq:bound}
\end{equation}
\label{thm:bound}
\end{theorem}

\begin{proof}
Suppose that one optimal solution pair to (\ref{eq:bus}) is $(\hat{\bfx},\hat{\bfz})$. Since $\|\hat{\bfz}\|_{2} \leq \epsilon$, we have $\|\hat{\bfz}\|_{1} \leq \sqrt{n}\|\hat{\bfz}\|_{2} \leq  \sqrt{n} \epsilon $.

Since $\x$ and $\bfz=\bfv$ are feasible for (\ref{eq:bus}) and $\y=H\x+\e+\bfv$, then
\begin{eqnarray*}
&&\|\y- H\hat{\x}-\hat{\z}\|_{1}\\
&=&\|H(\x-\hat{\x})+\e+\bfv-\hat{\z}\|_{1}\\
&\leq& \|H(\x-\x)+\e+\bfv-\bfv\|_{1}\\
&=&\|\e\|_{1}.
\end{eqnarray*}

Applying the triangle inequality to $\|H(\x-\hat{\x})+\e+\bfv-\hat{\z}\|_{1}$, we further obtain
\begin{equation*}
 \|H(\x-\hat{\x})+\e\|_{1}-\|\bfv\|_{1}-\|\hat{\bfz}\|_1 \leq \|\e\|_1.
\end{equation*}

Denoting $H(\x-\hat{\x})$ as $\w$, because $\e$ is supported on a set $K$ with cardinality $|K| \leq k$, by the triangle inequality for $\ell_1$ norm again,
\begin{equation*}
 \|\e\|_{1}-\|\w_{K}\|_1+\|\w_{\overline{K}}\|_1-\|\bfv\|_{1}-\|\hat{\z}\|_1\leq \|\e\|_1.
\end{equation*}

So we have
\begin{equation}
 -\|\w_{K}\|_1+\|\w_{\overline{K}}\|_1 \leq \|\hat{\z}\|_1+\|\bfv\|_{1} \leq 2\sqrt{n}\epsilon
\label{eq:differencebounded}
\end{equation}

With $C\|\w_K\|_{1} \leq \|\w_{\overline{K}}\|_{1}$, we know
\begin{equation*}
 \frac{C-1}{C+1} \|\w\|_{1} \leq -\|\w_{K}\|_1+\|\w_{\overline{K}}\|_1.
\end{equation*}

Combining this with (\ref{eq:differencebounded}), we obtain
\begin{equation*}
 \frac{C-1}{C+1} \|\w\|_{1} \leq 2\sqrt{n}\epsilon.
\end{equation*}

By the almost Euclidean property $\alpha \sqrt{n} \|\w\|_2 \leq \|\w\|_{1}$, it follows:
\begin{equation}
\|\w\|_{2} \leq \frac{2(C+1)}{\alpha (C-1)}\epsilon.
\label{eq:wl2norm}
\end{equation}

By the definition of singular values,
\begin{equation}
\sigma_{\text{min}} \|\x-\hat{\x}\|_2 \leq \|H(\x-\hat{\x})\|_2=\|\w\|_2,
\end{equation}
so combining (\ref{eq:wl2norm}), we get
\begin{equation*}
\|\x-\hat{\x}\|_{2} \leq \frac{2(C+1)}{\sigma_{\text{min}} \alpha (C-1)}\epsilon.
\end{equation*}
\end{proof}

Note that when there are no sparse errors present, the decoding error bound using the standard LS method satisfies $\|\x-\hat{\x}\|_{2} \leq \frac{1}{\sigma_{\text{min}} }\epsilon$ \cite{CandesErrorCorrection}. Theorem \ref{thm:bound} shows that the decoding error bound of (\ref{eq:bus}) is oblivious to the amplitudes of these bad data. This phenomenon was also observed in \cite{CandesErrorCorrection} by using the restricted isometry condition for compressive sensing.

We remark that, for given $\y$ and $\epsilon$, by strong Lagrange duality theory, the solution $\hat{\x}$ to (\ref{eq:bus}) corresponds to the solution to $\x$ in the following problem (\ref{eq:lambdaduality}) for some Lagrange dual variable $\lambda \geq 0$. 
\begin{equation}\label{eq:lambdaduality}
\min_{\bfx, \bfz} \quad \|\bfy-H\bfx-\bfz\|_1 +\lambda \|\bfz\|_2.
\end{equation}
In fact, when $\lambda \rightarrow \infty$, the optimizer $\|\z\|_2 \rightarrow 0$, and (\ref{eq:lambdaduality})
approaches
\begin{equation}\nonumber
\min_{\bfx} \quad \|\bfy-H\bfx\|_1,
\label{eq:appl1}
\end{equation}
and when $\lambda \rightarrow 0$, the optimizer $\z \rightarrow \y-H\x$, and (\ref{eq:lambdaduality}) approaches
\begin{equation}\nonumber
\min_{\bfx} \quad \|\bfy-H\bfx\|_2.
\end{equation}

 In the next two sections, we aim at explicitly computing
$\frac{2(C+1)}{\sigma_{\text{min}} \alpha (C-1)} \times \sqrt{n}$ appearing in the error bound (\ref{eq:bound}),
which is subsequently denoted as $\varpi$ in this paper. The appearance of
the $\sqrt{n}$ factor is to compensate for the energy scaling of
large random matrices and its meaning will be clear in later
context. We first compute explicitly the almost Euclidean
property constant $\alpha$, and then use the almost Euclidean property to get a direct estimate of the constant $C$ in the error bound (\ref{eq:bound}).

\section{Bounding the Almost Euclidean Property}
\label{sec:boundingEuclidean}
In this section, we would like to give a quantitative bound on the
almost Euclidean property constant $\alpha$ such that with high
probability (with respect to the measure for the subspace generated
by random $H$), $\alpha \sqrt{n} \|\w\|_2 \leq \|\w\|_{1}$ holds
for every vector $\w$ from the subspace generated by $H$. Here we
assume that each element of $H$ is generated from the standard
Gaussian distribution $N(0,1)$. Hence the subspace generated by $H$ is
a uniformly distributed $m$-dimensional subspace.

 To ensure that the subspace generated from $H$ satisfies the almost Euclidean property with $\alpha>0$, we must have the event that the subspace generated by $H$ does not intersect the set $\{\w\in \mathbb{S}^{n-1}| \|\w\|_{1}< \alpha \sqrt{n} \|\w\|_2=\alpha \sqrt{n}\}$, where $\mathbb{S}^{n-1}$ is the unit Euclidean sphere in $R^n$. To evaluate the probability that this event happens, we will need the following ``escape-through-mesh'' theorem.
\begin{theorem} [Escape through the mesh \cite{Gordon}]
\label{thm:escapethroughmesh}
Let $S$ be a subset of the unit Euclidean sphere $\mathbb{S}^{n-1}$ in $R^n$. Let $Y$ be a random $m$-dimensional subspace of $R^{n}$, distributed uniformly in the Grassmanian with respect to the Haar measure. Let us further take $w(S)$=$E(\sup_{\w\in S}$$(\bfh^T\bfw))$, where $\bfh$ is a random column vector in $R^{n}$ with i.i.d. $N(0,1)$ components. Assume that $w(S) < (\sqrt{n-m}-\frac{1}{2\sqrt{n-m}})$. Then
\begin{equation*}
P(Y \bigcap S=\emptyset)>1-3.5 e^{-\frac{(\sqrt{n-m}-\frac{1}{2\sqrt{n-m}})-w(S)}{18}}.
\end{equation*}
\end{theorem}

 We derive the following upper bound of $w(\bfh, S)=\sup_{\w\in S}(\bfh^T\bfw)$ for an arbitrary but fixed $\bfh$. Because the set $\{\w\in S^{n-1}| \|\w\|_{1}< \alpha \sqrt{n} \|\w\|_2 \}$ is symmetric, without loss of generality, we assume that the elements of $\bfh$ follow i.i.d.  half-normal distributions, namely the distribution for the absolute value of a standard zero mean Gaussian random variables. With $h_i$ denoting the $i$-th element of $\bfh$, $\sup_{\w\in S}(\bfh^T\bfw)$ is equivalent to
\begin{eqnarray}
\label{eq:maxproductoptimization}
\max && \sum_{i=1}^{n} h_{i} w_{i}\nonumber\\
{\text{subject to}}&& w_{i} \geq 0, 1\leq i \leq n\nonumber\\
&&\sum_{i=1}^{n}w_{i} \leq \alpha \sqrt{n}\nonumber\\
&&\sum_{i=1}^{n} w_{i}^2=1\nonumber.
\end{eqnarray}

Following the method from \cite{StojnicThresholds}, we use the Lagrange duality to find an upper bound for the objective function of (\ref{eq:maxproductoptimization}):
\begin{eqnarray}
&&\min_{u_{1} \geq 0, u_{2}\geq 0, \lambda \geq 0}\max_{w}\bfh^{T} \w-u_{1}(\sum_{i=1}^{n}w_i^2-1)\nonumber\\
&&-u_2(\sum_{i=1}^{n}w_{i}-\alpha \sqrt{n})+\sum_{i=1}^{n} \lambda_{i} w_{i},
\label{eq:maxmin}
\end{eqnarray}
where $\lambda$ is a vector $(\lambda_{1}, \lambda_{2}, ..., \lambda_{n})$. Note that restricting $u_1$ to be nonnegative still gives an upper bound even though it corresponds to an equality.

First, we maximize (\ref{eq:maxmin}) over $w_{i}$, $i=1,2, ..., n$ for fixed $u_{1}$, $u_{2}$ and $\lambda$. By setting the derivatives to be zero, the maximizing $w_{i}$ is given by
\begin{equation*}
w_{i}=\frac{h_{i}+\lambda_{i}-u_{2}}{2u_{1}},  1\leq i \leq n
\end{equation*}

Plugging this back into the objective function in (\ref{eq:maxmin}), we get
\begin{eqnarray}
 &&\bfh^{T} \w-u_{1}(\sum_{i=1}^{n}w_i^2-1)\nonumber\\
 &&-u_2(\sum_{i=1}^{n}w_{i}-\alpha \sqrt{n})+\sum_{i=1}^{n} \lambda_{i} w_{i}\nonumber\\
 &&=\frac{\sum_{i=1}^{n}{(-u_2+\lambda_{i}+h_i)^2}}{4 u_{1}}+u_1+\alpha \sqrt{n}u_2.
 \label{eq:insidemin}
\end{eqnarray}

Next, we minimize (\ref{eq:insidemin}) over $u_1 \geq 0$. It is not
hard to see the minimizing $u_{1}^*$ is
\begin{equation*}
 u_{1}^*=\frac{\sqrt{\sum_{i=1}^{n}{(-u_2+\lambda_{i}+h_i)^2}}}{2},
\end{equation*}
and the corresponding minimized value is
\begin{equation}\label{eqn:insidemin2}
{\sqrt{\sum_{i=1}^{n}{(-u_2+\lambda_{i}+h_i)^2}}}+\alpha \sqrt{n}
u_2.
\end{equation}

Then, we minimize (\ref{eqn:insidemin2}) over $\lambda \geq 0$. Given $\bfh$ and $u_2 \geq 0$, it is easy to see that the minimizing
$\lambda$ is

\[ \lambda_i = \left\{ \begin{array}{ll}
         u_2-h_i & \mbox{if $h_i \leq u_2$};\\
        0 & \mbox{otherwise},\end{array} \right. \]
and the corresponding minimized value is
\begin{equation}\label{eq:upperboundinstance}
\sqrt{\sum_{1 \leq i \leq n:\\ h_i >u_2}(u_2-h_i)^2}+\alpha
\sqrt{n}u_2.
\end{equation}

Now if we take any $u_2 \geq 0$, (\ref{eq:upperboundinstance}) serves
as an upper bound for (\ref{eq:maxmin}), and thus also an upper bound for $\sup_{\w\in S}(\bfh^T\bfw)$. Since $\sqrt{\cdot}$ is a
concave function, by Jensen's inequality, we have for any given $u_2
\geq 0$,
\begin{equation}\label{eqn:upper}
E(\sup_{\w\in S}(\bfh^Tw)) \leq \sqrt{ E\{\sum_{1\leq i \leq n: h_i
>u_2}{(u_2-h_i)^2}\} }+\alpha \sqrt{n} u_2.
\end{equation}
Since $\bfh$ has i.i.d. half-normal components, the righthand
side of (\ref{eqn:upper}) equals to
\begin{equation}\label{eqn:boundu2}
(\sqrt{(u_2^2+1)\textrm{erfc}(u_2/\sqrt{2})-\sqrt{2/\pi}u_2e^{-u_2^2/2}}+\alpha
u_2)\sqrt{n},
\end{equation}
where $\textrm{erfc}$ is the complementary error function.

One can check that (\ref{eqn:boundu2}) is convex in $u_2$. Given
$\alpha$, we minimize (\ref{eqn:boundu2}) over $u_2 \geq 0$ and let
$g(\alpha)\sqrt{n}$ denote the minimum value. Then from
(\ref{eqn:upper}) and (\ref{eqn:boundu2}) we know
\begin{equation}
w(S)=E(\sup_{\w\in S}(\bfh^T\w)) \leq g(\alpha) \sqrt{n}.
\end{equation}
Given $\delta=\frac{m}{n}$, we pick the largest $\alpha^*$ such that
$g(\alpha^*) <\sqrt{1-\delta}$. Then as $n$ goes to infinity, it
holds that
\begin{equation}
w(S)\leq g(\alpha^*)\sqrt{n}<(\sqrt{n-m}-\frac{1}{2\sqrt{n-m}}).
\end{equation}
Then from Theorem \ref{thm:escapethroughmesh}, with high
probability  $\|\w\|_1 \geq \alpha^*\sqrt{n}\|\w\|_2$ holds for
every vector $\w$ in the subspace generated by $H$. We numerically
calculate how $\alpha^*$ changes over $\delta$ and plot the curve in
Fig. \ref{fig:alpha}. For example, when $\delta=0.5$,
$\alpha^*=0.332$, thus $\|\w\|_1 \geq 0.332\sqrt{n}\|\w\|_2$ for all
$\w$ in the subspace generated by $H$.

\begin{figure}[t] \centering
\includegraphics[scale=0.4]{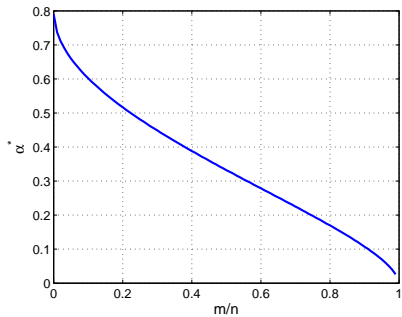}
\caption{$\alpha^*$ over $m/n$}\label{fig:alpha}
\end{figure}

Note that when $\frac{m}{n}=\frac{1}{2}$, we get $\alpha=0.332$. That is much larger than the known $\alpha$ used in \cite{Yin}, which is approximately $0.07$ (see Equation (12) in \cite{Yin}). When applied to the sparse recovery problem considered in \cite{Yin}, we are able to recover any vector with no more than $0.0289n=0.0578m$ nonzero elements, which are $20$ times more than the $\frac{1}{384}m$ bound in \cite{Yin}.

 \section{Evaluating the Robust Error Correction Bound and Comparisons with other Bounds}
\label{sec:evaluating}
If the elements in the measurement matrix $H$ are i.i.d. drawn from the Gaussian distribution $N(0,1)$, following upon the work of
Marchenko and Pastur \cite{Marcenko67}, Geman \cite{Geman80} and
Silverstein \cite{silver} proved that for $m/n=\delta$, as $n
\rightarrow \infty$, the smallest nonzero singular value
\begin{equation*}
\frac{1}{\sqrt{n}}\sigma_{\text{min}} \rightarrow 1-\sqrt{\delta}
\end{equation*}
almost surely as $n \rightarrow \infty$.

Now that we have already explicitly bounded $\alpha$ and $\sigma_{\text{min}}$, we now proceed to characterize $C$. It turns out that our earlier result on the almost Euclidean property can be used to compute $C$.
\begin{lemma}
 Suppose an $n$-dimensional vector $\w$ satisfies $\|\w\|_{1}\geq \alpha \sqrt{n} \|\w\|_2$, and for some set $K \subseteq \{1,2,...,n\}$ with cardinality $|K|=k \leq n$,
$\frac{\|\w_{K}\|_1}{\|\w\|_{1}}= \beta$.
Then $\beta$ satisfies
\begin{equation*}
\frac{\beta^2}{k}+\frac{(1-\beta)^2}{n-k} \leq \frac{1}{\alpha^2 n}.
\end{equation*}
\label{thm:almostEuclideanB}
\end{lemma}
\begin{proof}
Without loss of generality, we let $\|\w\|_1=1$. Then by the Cauchy-Schwarz inequality,
\begin{eqnarray*}
\|\w\|_{2}^2&=&\|\w_{K}\|_{2}^2+\|\w_{\overline{K}}\|_{2}^2\\
&\geq& (\frac{\|\w_{K}\|_{1}}{\sqrt{k}})^2 + (\frac{\|\w_{\overline{K}}\|_{1}}{\sqrt{n-k}})^2\\
&=& (\frac{\beta^2}{k}+\frac{(1-\beta)^2}{n-k})\|\w\|_{1}^2.
\end{eqnarray*}

At the same time, by the almost Euclidean property,
\begin{equation*}
\alpha^2 n \|\w\|_{2}^2 \leq \|\w\|_{1}^2,
\end{equation*}
so we must have
\begin{equation*}
\frac{\beta^2}{k}+\frac{(1-\beta)^2}{n-k} \leq \frac{1}{\alpha^2 n}.
\end{equation*}
\end{proof}

\begin{corollary}
If a nonzero $n$-dimensional vector $\w$ satisfies $\|\w\|_{1}\geq \alpha \sqrt{n} \|\w\|_2$, and if for any set $K \subseteq \{1,2,...,n\}$ with cardinality $|K|=k \leq n$, $C\|\w_{K}\|_{1}= \|\w_{\overline{K}}\|_1$ for some number $C\geq 1$, then
\begin{equation*}
\frac{1}{\frac{k}{n}}+\frac{C^2}{1-\frac{k}{n}} \leq \frac{(C+1)^2}{\alpha^2}.
\end{equation*}
\label{corollary:Cbound}
\end{corollary}

\begin{proof}
If $C\|\w_{K}\|_{1}= \|\w_{\overline{K}}\|_1$, we have
\begin{equation*}
\frac{\|\w_{K}\|_1}{\|\w\|_{1}}= \frac{1}{C+1}.
\end{equation*}
So by Lemma \ref{thm:almostEuclideanB}, $\beta=\frac{1}{C+1}$ satisfies
\begin{equation*}
\frac{\beta^2}{k}+\frac{(1-\beta)^2}{n-k} \leq \frac{1}{\alpha^2 n}.
\end{equation*}

This is equivalent to
\begin{equation*}
\frac{1}{\frac{k}{n}}+\frac{C^2}{1-\frac{k}{n}} \leq \frac{(C+1)^2}{\alpha^2}.
\end{equation*}

\end{proof}

\begin{corollary}
Let $\y$, $\x$, $H$, $\e$ and $\bfv$ be specified as above. Assume that $H$ is drawn i.i.d. Gaussian distribution $N(0,1)$ and the subspace generated by $H$ satisfies the \emph{almost Euclidean} property for a constant $\alpha \leq 1$ with overwhelming probability as $n \rightarrow \infty$.
Then for any small number $\gamma>0$, almost surely as $n \rightarrow \infty$, \emph{simultaneously} for any state $\x$ and for any sparse error $\e$ with at most $k$ nonzero elements, the solution $\hat{\x}$ to (\ref{eq:bus}) satisfies
\begin{equation}
\|\x-\hat{\x}\|_{2} \leq \frac{2(1+\gamma)(C+1)}{(\sqrt{n}-\sqrt{m})\alpha (C-1)}  \epsilon,
\label{eq:computedbound}
\end{equation}
where $C$ is the smallest nonnegative number such that $\frac{1}{\frac{k}{n}}+\frac{C^2}{1-\frac{k}{n}} \leq \frac{(C+1)^2}{\alpha^2}$ holds; and $\frac{k}{n}$ needs to be small enough such that $C>1$.
\label{corollary:boundexplicit}
\end{corollary}

\begin{proof}
This follows from \ref{corollary:Cbound} and $\frac{1}{\sqrt{n}}\sigma_{\text{min}} \rightarrow 1-\sqrt{\delta}$ as $n\rightarrow \infty$.
\end{proof}

So for a sparsity ratio $\frac{k}{n}$, we can use the procedure in Section \ref{sec:boundingEuclidean} to calculate the constant $\alpha$ for the almost Euclidean property. Then we can use $\alpha$ to find the value for $C$ in Corollary \ref{corollary:boundexplicit}. In Figure \ref{fig:varpi}, we plot $\frac{2(C+1)}{\sigma_{\text{min}} \alpha (C-1)} \sqrt{n}=\varpi$ appearing in Corollary \ref{corollary:boundexplicit} for $\delta=\frac{m}{n}=\frac{1}{2}$ as a function $\frac{k}{n}$. Apparently, when the sparsity $\frac{k}{n}$ increases, the recovery error bound also increases.

\begin{figure}[t]
\centering
\includegraphics[scale=0.4]{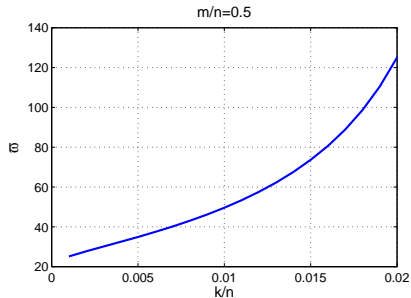}
\caption{$\varpi$ versus $\frac{k}{n}$}\label{fig:varpi}
\end{figure}

\subsection{Comparisons with Existing Performance Bounds}
In this subsection, we would like to explore the connection between robust sparse error correction and compressive sensing, and compare our performance bound with the bounds in the compressive sensing literature.

As already noted in \cite{CandesErrorCorrection} and \cite{CT1}, sparse error correction can be seen as a dual to the compressive sensing problem. If we multiply $\y=H\x+\e+\mathbf{v}$ with an $(n-m)\times n$ matrix $A$ satisfying $A\times H=0$,
\begin{equation*}
A\y=AH\x+A\e+A\mathbf{v}=A\e+A\mathbf{v}.
\end{equation*}
Since $\e$ is a sparse vector, this corresponds to a compressive sensing problem with sensing matrix $A$, observation result $A\y$ and observation noise $A\mathbf{v}$.

\cite{Neighborlypolytope} and \cite{DonohoTanner} used the theory of high dimensional convex geometry to establish the phase transition on the recoverable sparsity for perfectly sparse signals under noiseless observations. Compared with \cite{Neighborlypolytope} and \cite{DonohoTanner}, our method using the almost Euclidean property applies to the setting of noisy observations. \cite{XuHassibi} also used the convex geometry tool to get the precise stability phase transition for compressive sensing. However, \cite{XuHassibi} mainly focused on obtaining $\ell_1$-$\ell_1$ error bounds, namely the recovery error measured in $\ell_1$ norm is bounded in terms of the $\ell_1$ norm of perturbations. No explicitly computed error bound measured in $\ell_2$ norm was given in \cite{XuHassibi}.

In a remarkable paper \cite{DonohoMalekiMontanari}, precise noise-sensitivity phase transition was provided for compressive sensing problems with Gaussian observation noises. In \cite{DonohoMalekiMontanari}, expected recovery error in $\ell_2$ norm was considered and the average was taken over the signal input distribution and the Gaussian observation noise.
Compared with \cite{DonohoMalekiMontanari}, this paper derives a \emph{worst-case }bound on the recovery error which holds true for \emph{all} $k$-sparse vectors \emph{simultaneously}; the performance bound also applies to arbitrary observation noises, including but not limited to Gaussian observation noise. Compared with another worst-case recovery error bound obtained in \cite{CandesErrorCorrection} through restricted isometry condition, the bound in this paper greatly improves on the sparsity up to which robust error correction happens. To our best knowledge, currently the best bound on the sparsity for a restricted isometry to hold is still very small. For example, when $\frac{n-m}{n}=\frac{1}{2}$, the proven sparsity $\frac{k}{n}$ for strong $\ell_0$-$\ell_1$ equivalence is about $1.5 \times 10^{-3}$ according to \cite{blanchard}, which is smaller than the bound $\frac{k}{n}=0.0289$ obtained in this paper. In fact, as illustrated in Figure \ref{fig:alpha}, the almost Euclidean bound exists for any $\frac{m}{n} <1$. So combined with the recovery error bound in $\ell_1$ norm, we can show the recovery error is bounded in $\ell_2$ norm up to the stability sparsity threshold in \cite{XuHassibi}. In this paper, we also obtain much sharper bounds on the almost Euclidean property than known in the literature \cite{Yin}.

\section{Sparse Error Correction from Nonlinear Measurements}
\label{sec:nonlinear}
In applications, measurement outcome can be nonlinear functions of system states. Let us denote the $i$-th measurement by $h_{i}(\x)$, where $1\leq i\leq n$ and $h_{i}(\x)$ can be a nonlinear function of $\x$. In this section, we study the theoretical performance guarantee of sparse recovery from nonlinear measurements and give an iterative algorithm to do sparse recovery from nonlinear measurements, for which we provide conditions under which the iterative algorithm converges to the true state. 

In Subsection \ref{subsec:theory}, we explore the conditions under which sparse recovery from nonlinear measurements are theoretically possible. In Subsection \ref{subsec:algorithmdescription}, we describe our iterative algorithm to perform sparse recovery from nonlinear measurements. In Subsection \ref{subsec:algorithm}, we study the algorithm performance guarantees when the measurements are with or without additive noise.

\subsection{Theoretical Guarantee for Direct $\ell_0$ and $\ell_1$-Minimization}
\label{subsec:theory}
We first give a general condition which guarantees recovering correctly the state $\x$ from the corrupted observation $\y$ without considering the computational cost.

\begin{theorem}
Let $h(\cdot)$, $\x$, and $\e$ be specified as above; and $\y=h(\x)+\e$. A state $\x$ can be recovered correctly from any error $\e$ with $\|\e\|_0 \leq k$ from solving the optimization
 \begin{equation}\label{eqn:norm0}
\min_{\bfx} \quad \|\bfy-h(\bfx)\|_0,
\end{equation}
if and only if for any $\x^* \neq \x$, $\|h(\x)-h(\x^*)\|_{0} \geq 2k+1$.
\label{thm:generalnonlinear}
\end{theorem}

\begin{proof}
We first prove the sufficiency part, namely if for any $\x^* \neq \x$, $\|h(\x)-h(\x^*)\|_{0} \geq 2k+1$, we can always correctly recover $\x$ from $\y$ corrupted with any error $\e$ with $\|\e\|_0 \leq k$.  Suppose that instead an solution to the optimization problem (\ref{eqn:norm0}) is an $\x^* \neq \x$.
Then
\begin{eqnarray*}
&&\|\y- h(\x^*)\|_{0}\\
&=&\|(h(\x)+\e)- h(\x^*)\|_{0}\\
&\geq& \|h(\x)- h(\x^*)\|_{0}-\|\e\|_{0}\\
&\geq& (2k+1)-k\\
&>& \|\e\|_0=\|\y- h(\x)\|_{0}.
\end{eqnarray*}

So $\x^* \neq \x$ can not be a solution to (\ref{eqn:norm0}), which is a contradiction.

For the necessary part, suppose that there exists an $\x^* \neq \x$ such that $\|h(\x)-h(\x^*)\|_{0} \leq 2k$.  Let $I$ be the index set where $h(\x)$ and $h(\x^*)$ differ and its size $|I| \leq 2k$. Let $\gamma=h(\x^*)-h(\x)$.  We pick $\e$ such that $\e_i=\gamma_i$, $\forall i \in I'$, where $I'\subseteq I$ is an index set with cardinality $|I'|=k$; and $\e_i$ to be $0$ otherwise. Then
\begin{eqnarray*}
&&\|\y- h(\x^*)\|_{0}\\
&=& \|h(\x)- h(\x^*)+\e\|_{0}\\
&=& |I|-k\\
&\leq& k=\|\e\|_0=\|\y-h(\x)\|_0,
\end{eqnarray*}
which means that $\x$ can not be a solution to (\ref{eqn:norm0}) and is certainly not a unique solution to (\ref{eqn:norm0}).
\end{proof}

\begin{theorem}
Let $h(\cdot)$, $\x$, and $\e$ be specified as above; and $\y=h(\x)+\e$. A state $\x$ can be recovered correctly from any error $\e$ with $\|\e\|_0 \leq k$ from solving the optimization
 \begin{equation}\label{eqn:norm1}
\min_{\bfx} \quad \|\bfy-h(\bfx)\|_1,
\end{equation}
 if and only if for any $\x^* \neq \x$, $\|(h(\x)-h(\x^*))_{K}\|_{1} < \|(h(\x)-h(\x^*))_{\overline{K}}\|_{1}$, where $K$ is the support of the error vector $\e$.
\label{thm:l1nonlinear}

\end{theorem}

\begin{proof}
We first prove  if any $\x^* \neq \x$, $\|(h(\x)-h(\x^*))_{K}\|_{1} < \|(h(\x)-h(\x^*))_{\overline{K}}\|_{1}$, where $K$ is the support of the error vector $\e$, we can correctly recover state $\x$ from (\ref{eqn:norm1}). Suppose that instead an solution to the optimization problem (\ref{eqn:norm0}) is an $\x^* \neq \x$.
Then
\begin{eqnarray*}
&&\|\y- h(\x^*)\|_{1}\\
&=&\|(h(\x)+\e)- h(\x^*)\|_{1}\\
&=& \|\e_{K}- (h(\x^*)-h(\x))_{K}\|_{1}+\|(h(\x^*)-h(\x))_{\overline{K}}\|_{1}\\
&\geq& \|\e_{K}\|_{1}- \|(h(\x^*)-h(\x))_{K}\|_{1}+\|(h(\x^*)-h(\x))_{\overline{K}}\|_{1}\\
&>& \|\e_{K}\|_{1}=\|\y-h(\x)\|_1.
\end{eqnarray*}
So $\x^* \neq \x$ can not be a solution to (\ref{eqn:norm1}), and this leads to a contradiction.

 Now suppose that there exists an $\x^* \neq \x$ such that $\|(h(\x)-h(\x^*))_{K}\|_{1} \geq \|(h(\x)-h(\x^*))_{\overline{K}}\|_{1}$, where $K$ is the support of the error vector $\e$.
Then we can pick $\e$ to be $(h(\x^*)-h(\x))_{K}$ over its support $K$ and to be $0$ over $\overline{K}$. Then
\begin{eqnarray*}
&&\|\y- h(\x^*)\|_{1}\\
&=& \|h(\x)- h(\x^*)+\e\|_{1}\\
&=& \|(h(\x^*)-h(\x))_{\overline{K}}\|_{1}\\
&\leq& \|(h(\x)-h(\x^*))_{K}\|_{1}=\|\e\|_1=\|y-h(\x)\|_1,
\end{eqnarray*}
which means that $\x$ can not be a solution to (\ref{eqn:norm1}) and is certainly not a unique solution to (\ref{eqn:norm1}).
\end{proof}

However, direct $\ell_0$ and $\ell_1$ minimization may be computationally costly because $\ell_0$ norm and nonlinear $h(\cdot)$ may lead to non-convex optimization problems. In the next subsection, we introduce our computationally efficient iterative sparse recovery algorithm in the general setting when the additive noise $\bfv$ is present.

\subsection{Iterative $\ell_1$-Minimization Algorithm}
\label{subsec:algorithmdescription}
Let $h(\cdot)$, $\x$, $\e$ and $\bfv$ be specified as above; and $\y=h(\x)+\e+\bfv$ with $\|\bfv\|_2\leq \epsilon$. Now let us consider the algorithm which recovers the state variables iteratively. Ideally, an estimate of the state variables, $\hat{\bfx}$, can be obtained by solving the following minimization problem,
\begin{eqnarray}\label{eqn:nonlinear}
\min_{\bfx,\z} &&\|\bfy-h(\bfx)-\z\|_1,\nonumber\\
{\text{subject to}}&&\|\z\|_2 \leq \epsilon.
\end{eqnarray}
where $\hat{\bfx}$ is the optimal solution $\bfx$. Even though the $\ell_1$ norm is a convex function,  the function $h(\cdot)$ may make the objective function non-convex.

Since $h$ is nonlinear, we linearize the equations and apply an
iterative procedure to obtain a solution. We start with an initial
state $\bfx^0$. In the $k$-th ($k \geq 1$) iteration, let $\Delta \bfy^k=\bfy-h(\bfx^{k-1})$, then we solve the following
convex optimization problem,
\begin{eqnarray}\label{eqn:linear}
\min_{\Delta \x,\z}  &&\|\Delta \y^{k}-H^{local}\Delta\x-\z\|_{1},\nonumber \\
{\text{subject to}}&&\|\z\|_2 \leq \epsilon,
\end{eqnarray}
where $H^{local}$ is the $n \times m$ Jacobian matrix of $h$ evaluated at the point
$\bfx^{k-1}$. Let $\Delta\bfx^k$ denote the optimal solution
$\Delta\bfx$ to (\ref{eqn:linear}), then the state estimation is
updated by
\begin{equation}
\bfx^{k}=\bfx^{k-1}+\Delta\bfx^k.
\end{equation}
We repeat the process until $\Delta \bfx^k$ approaches $0$ close enough or $k$ reaches a specified maximum value.

Note that when there is no additive noise, we can take $\epsilon=0$ in this iterative algorithm. When there is no additive noise, the algorithm is exactly the same as the state estimation algorithm from \cite{KV82}.

\subsection{Convergence Conditions for the Iterative Sparse Recovery Algorithm}
\label{subsec:algorithm}

In this subsection, we discuss the convergence of the proposed algorithm in Subsection \ref{subsec:algorithmdescription}. First, we give a necessary condition (Theorem \ref{thm:l1local}) for recovering the true state when there is no additive noise, and then give a sufficient condition (Corollary \ref{thm:l1iterative}) for the iterative algorithm to converge to the true state in the absence of additive noise. Secondly, we give the performance bounds (Theorem \ref{thm:l1iterativenoise}) for the iterative sparse error correction algorithm when there is additive noise.

\begin{theorem}[Necessary Recovery Condition]
Let $h(\cdot)$, $\x$, and $\e$ be specified as above; and $\y=h(\x)+\e$. The iterative algorithm converges to the true state $\x$
only if for the Jacobian matrix $H^{local}$ at the point of $\x$ and for any $\x^*\neq 0$, $\|{(H^{local}\x^*)}_{K}\|_1 > \|{(H^{local}\x^*)}_{\overline{K}}\|_1$, where $K$ is the support of
the error vector $\e$.
\label{thm:l1local}
\end{theorem}

\begin{proof}
The proof follows from the proof for Theorem \ref{thm:l1nonlinear}, with the linear function $g(\Delta \x)=h(\x)+H^{local}\Delta \x$, where $H^{local}$ is the Jacobian matrix at the true state $\x$.
\end{proof}

Theorem \ref{thm:l1local} shows that for nonlinear measurements, the local Jacobian matrix needs to satisfy the same condition as the matrix for linear measurements. This assumes that the iterative algorithm starts with the correct initial state. However, the iterative algorithm generally does not start at the true state $\x$. In the following theorem, we give a sufficient condition for the algorithm to have an upper bounded estimation error when additive noises are present, even though the starting point is not precisely the same one. As a corollary of this theorem, the proposed algorithm converges to the true state when there is no additive noise.

\begin{theorem}[Guarantee with Additive noise]
Let $h(\cdot)$, $\x$, $\e$, and $n$ be specified as above; and $\y=h(\x)+\e+\bfv$ with $\|\bfv\|_2\leq \epsilon$. Suppose that at every point $\x$, the local Jacobian matrix $H^{local}$ is full rank and satisfies that for every $\z$ in the range of $H^{local}$, $C\|\z_{K}\|_1 \leq \|\z_{\overline{K}}\|_1$, where $K$ is the support of the error vector $\e$ and $C$ a constant larger than $1$. Moreover, for a fixed constant $\beta<1$, we assume that
\begin{equation}\label{eqn:optexpwithnoise}
\frac{2(C+1)}{C-1} \frac{\sigma_{max}^{1}(H^{true}-H^{local})}{\sigma_{min}^{1}(H^{local})}\leq\beta
\end{equation}
holds for any two states $\x_1$ and $\x_2$, where $H^{local}$ is the local Jacobian matrix at the point $\x_1$, $H^{true}$ is a matrix such that $h(\x_2)-h(\x_1)=H^{true} (\x_2-\x_1)$, $\sigma_{max}^{1}(A)$ is the induced $\ell_1$ matrix norm for $A$, and $\sigma_{min}^{1}(A)$ for a matrix $A$ is defined as $ \sigma_{min}^{1}(A)= \min\{\|A\z\|_1 : \mbox{ with }\|\z\|_1= 1\}$.

Then for any true state $\x$, the estimation $\x^{k+1}=\x^{k}+\Delta x^{k+1}$, where $\Delta x^{k+1}$ is the solution to the $(k+1)$-th iteration optimization
\begin{eqnarray}
\min_{\Delta \x^{k+1},\z}  &&\|\Delta \y^{k+1}-H^{local}\Delta\x^{k+1}-\z\|_{1},\nonumber \\
{\text{subject to}}&&\|\z\|_2 \leq \epsilon
\label{eq:optheoremnoise}
\end{eqnarray}
satisfies
\begin{eqnarray}\label{eqn:optexp7noise}
\|\x&-&\x^{k+1}\|_1 \leq \frac{2(C+1)}{(C-1){\sigma_{min}^{1}(H^{local})}}\times 2\sqrt{n} \epsilon\nonumber\\
&+&\frac{2(C+1)}{C-1} \frac{\sigma_{max}^{1}(H^{true}-H^{local})}{\sigma_{min}^{1}(H^{local})}\|\x-\x^{k}\|_1.\nonumber
\end{eqnarray}
As $k\rightarrow \infty$, with $\frac{2(C+1)}{C-1} \frac{\sigma_{max}^{1}(H^{true}-H^{local})}{\sigma_{min}^{1}(H^{local})}\leq\beta<1$,
\begin{eqnarray}\label{eqn:optexp8noise}
\|\x-\x^{k+1}\|_1 \leq \frac{2(C+1)}{(1-\beta)(C-1){\sigma_{min}^{1}(H^{local})}}\times 2\sqrt{n} \epsilon.\nonumber
\end{eqnarray}
\label{thm:l1iterativenoise}
\end{theorem}

\textbf{Remarks}: When the function is linear and therefore $H^{true}=H^{local}$, the condition (\ref{eqn:optexpwithnoise}) will always be satisfied with $\beta=0$. So nonlinearity is captured by the the term $\sigma_{max}^{1}(H^{true}-H^{local})$.

\begin{proof}
At the $k$-th iteration of the iterative state estimation algorithm
\begin{equation}\label{eqn:expansionnoise}
\y=H^{true}\Delta\bfx^*+h(\x^{k})+\e+\bfv,
\end{equation}
where $H^{true}$ is an $n \times m$ matrix, $\x^{k}$ is the state estimate at the $k$-th step, and $\Delta\bfx^*=\x-\x^k$, namely the estimation error at the $k$-th step.

Since at the $(k+1)$-th step, we are solving the following optimization problem
\begin{eqnarray}\label{eqn:optexp1noise}
\min_{\Delta \x,\z}  &&\|\Delta \y^{k+1}-H^{local}\Delta\x-\z\|_{1},\nonumber \\
{\text{subject to}}&&\|\z\|_2 \leq \epsilon.
\end{eqnarray}

Plugging (\ref{eqn:expansionnoise}) into (\ref{eqn:optexp1noise}),  we are really solving
\begin{eqnarray}\label{eqn:optexp2noise}
\min_{\Delta \x,\z}  &&\|H^{true}\Delta\x^{*}+\e+\bfv-H^{local}\Delta \x-\z\|_{1},\nonumber \\
{\text{subject to}}&&\|\z\|_2 \leq \epsilon.
\end{eqnarray}

Denoting $(H^{true}-H^{local})\Delta\x^{*}$ as $\w$, which is the measurement gap generated by using the local Jacobian matrix $H^{local}$ instead of $H^{true}$, then (\ref{eqn:optexp2noise}) is equivalent to
%
%
%
\begin{eqnarray}\label{eqn:optexp3noise}
&\min_{\Delta \x,\z}& \|H^{local}(\Delta\x^{*}-\Delta\x)+ \w+\e+\bfv-\z\|_{1},\nonumber \\
&{\text{subject to}}& \|\z\|_2 \leq \epsilon.
\end{eqnarray}

Suppose that the solution to (\ref{eq:optheoremnoise}) is $\Delta \x=\Delta \x^*-error$. We are minimizing the objective $\ell_1$ norm, and $(\Delta \x^{*},\bfv)$ is a feasible solution with an objective function value $\|\w+\e\|_1$, so we have
\begin{equation}\label{eqn:optexp4noise}
\|H^{local}\times error +\w+\e+\bfv-\z\|_1 \leq \|\w+\e\|_1.
\end{equation}

By the triangular inequality and the property of $H^{local}$, using the same line of reasoning as in the proof of Theorem \ref{thm:bound}, we have
\begin{eqnarray}\label{eqn:optexp5noise}
\|\e\|_1&+&\frac{C-1}{C+1}\|H^{local}\times error\|_1-\|\w\|_1-\|\bfv\|_1-\|\z\|_1 \nonumber\\
 &\leq& \|\e\|_1+\|\w\|_1.
\end{eqnarray}

So
\begin{equation}\label{eqn:optexp6noise}
\|H^{local}\times error\|_1 \leq \frac{2(C+1)}{C-1} (\|\w\|_1+\|\bfv\|_1+\|\z\|_1).
\end{equation}

Since $\|\bfv\|_1$ and $\|\z\|_1$ are both no larger than $2\sqrt{n}\epsilon$, $error=\Delta \x^*-\Delta \x$, $(\x-\x^k)=\Delta \x^{*}$, and $\x-\x^{k+1}=(\x-\x^k)-(\x^{k+1}-\x^k)=\Delta \x^*-\Delta\x$, we have
\begin{eqnarray}\label{eqn:optexp7noiseinproof}
\|\x&-&\x^{k+1}\|_1 \leq \frac{2(C+1)}{(C-1){\sigma_{min}^{1}(H^{local})}}\times 2\sqrt{n} \epsilon\nonumber\\
&+&\frac{2(C+1)}{C-1} \frac{\sigma_{max}^{1}(H^{true}-H^{local})}{\sigma_{min}^{1}(H^{local})}\|\x-\x^{k}\|_1,\nonumber
\end{eqnarray}
where $\sigma_{max}^{1}(H^{true}-H^{local})$ and $\sigma_{min}^{1}(H^{local})$ are respectively the matrix quantities defined in the statement of the theorem.

So as long as $\frac{2(C+1)}{C-1} \frac{\sigma_{max}^{1}(H^{true}-H^{local})}{\sigma_{min}^{1}(H^{local})}\leq\beta$,
for some fixed constant $\beta<1$, the error upper bound converges to $\frac{2(C+1)}{(1-\beta)(C-1){\sigma_{min}^{1}(H^{local})}}\times 2\sqrt{n} \epsilon$.
\end{proof}

 When there is no additive noise, as a corollary of Theorem \ref{thm:l1iterativenoise}, we know the algorithm converges to the true state.
\begin{corollary}[Correct Recovery without Additive noise]
Let $\y$, $h(\cdot)$, $\x$, $H$, and $\e$  be specified as above; and $\y=h(\x)+\e$. Suppose that at every point $\x$, the local Jacobian matrix $H^{local}$ is full rank and satisfies that for every $\z$ in the range of $H^{local}$, $C\|\z_{K}\|_1 \leq \|\z_{\overline{K}}\|_1$, where $K$ is the support of the error vector $\e$. Moreover, for a fixed constant $\beta<1$, we assume that \begin{equation}\label{eqn:optexp9}
\frac{2(C+1)}{C-1} \frac{\sigma_{max}^{1}(H^{true}-H^{local})}{\sigma_{min}^{1}(H^{local})}\leq\beta,
\end{equation}
holds true for any two states $\x_1$ and $\x_2$, where $H^{local}$ is the local Jacobian matrix at the point $\x_1$, $H^{true}$ is a matrix such that $h(\x_2)-h(\x_1)=H^{true} (\x_2-\x_1)$, $\sigma_{max}^{1}(A)$ is the induced $\ell_1$ matrix norm for $A$, and $\sigma_{min}^{1}(A)$ for a matrix $A$ is defined as $ \sigma_{min}^{1}(A)= \min\{\|A\z\|_1 : \mbox{ with }\|\z\|_1= 1\}$.

Then any state $\x$ can be recovered correctly from the observation $\y$  from the iterative algorithm in Subsection \ref{subsec:algorithmdescription}, regardless of the initial starting state of the algorithm.
\label{thm:l1iterative}
\end{corollary}

\subsection{Verifying the Conditions for Nonlinear Sparse Error Correction}
\label{subsec:verification}
To verify our conditions for sparse error correction from nonlinear measurements, we need to verify:
\begin{itemize}
\item at every iteration of the algorithm, for every $\z$ in the range of the local Jacobian matrix $H^{local}$, $C\|\z_{K}\|_1 \leq \|\z_{\overline{K}}\|_1$, where $K$ is the support of the error vector $\e$ and $C>1$ is a constant.
\item for some constant $\beta<1$,
\begin{equation}\label{eq:verification}
\frac{2(C+1)}{C-1} \frac{\sigma_{max}^{1}(H^{true}-H^{local})}{\sigma_{min}^{1}(H^{local})}\leq\beta
\end{equation}
\end{itemize}

To simplify our analysis, we only focus on verifying these conditions under which the algorithm converges to the true state when there is no observation noise (namely Corollary \ref{thm:l1iterative}), even though similar verifications also apply to the noisy setting (Theorem \ref{thm:l1iterativenoise}).

\subsubsection{Verifying the balancedness condition}

In nonlinear setting, the first condition, namely the balancedness condition \cite{XuHassibi}, needs to hold \emph{simultaneously} for many (or even an infinite number of) local Jacobian matrices $H^{local}$, while in the linear setting, the condition only needs to hold for a single linear subspace. The biggest challenge is then to verify that the balancedness condition holds \emph{simultaneously} for a set of linear subspaces. In fact,  verifying the balancedness condition for a single linear space is already very difficult \cite{Aspremont,Juditsky, Tangverify}. In this subsection, we show how to convert checking this balancedness condition into a semidefinite optimization problem, through which we can explicitly verify the conditions for nonlinear sparse error corrections.

One can verify the balancedness condition for each individual local Jacobian matrix $H^{local}$ in an instance of algorithm execution, however, the results are not general enough: our state estimation algorithm may pass through different trajectories of $H^{local}$'s, with different algorithm inputs and initializations. Instead, we propose to check the balancedness condition simultaneously for all $H^{local}$'s in a neighborhood of $H_0$, where $H_0$ is the local Jacobian matrix at the true system state. Once the algorithm gets into this neighborhood in which the balancedness condition holds for every local Jacobian matrix, the algorithm is guaranteed to converge to the true system state.

More specifically, we consider a neighborhood of the true system state, where each $H^{local}=H_0+\Delta H$. We assume that for each state in this neighborhood, each row of $\Delta H$ is bounded by $\epsilon>0$ in $\ell_2$ norm. The question is whether, for every $H^{local}$, every vector $\z$ in the range of $H^{local}$ satisfies $C\|\z_{K}\|_1 \leq \|\z_{\overline{K}}\|_1$,  for all subsets $K$ with cardinality $|K|\leq k$ ($k$ is the sparsity of bad data). This requirement is satisfied, if the optimal objective value of the following optimization problem is no more than $\alpha=\frac{1}{1+C}$:
\begin{eqnarray}
\max_{K} \max_{H} \max_{\x}  &&\|(H\x)_{K}\|_{1},\nonumber \\
{\text{subject to}}&&\|(H\x)_{\overline{K}}\|_1\leq 1 .
\label{eq:general bound}
\end{eqnarray}

However, this is a difficult problem to solve, because there are at least $\binom{n}{k}$ subsets $K$, and the objective is not a concave function. We instead solve a series of $n$ optimization problems to find an upper bound on the optimal objective value of (\ref{eq:general bound}). For each integer $1 \leq i\leq n$, by taking the set $K=\{i\}$, we solve
\begin{eqnarray}
\max_{H} \max_{\x}  && \|(H\x)_{\{i\}}\|_{1},\nonumber \\
{\text{subject to}}&&\|(H\x)_{\overline{\{i\}}}\|_1\leq 1.
\label{eq:tractable_bound_fork_fundamental}
\end{eqnarray}

Suppose the optimal objective value of (\ref{eq:tractable_bound_fork_fundamental}) is given by $\alpha_i$ for $1\leq i \leq n$. Then $\frac{\alpha_{i}}{1+\alpha_{i}}$ is the largest fraction $\|(H\x)_{\{i\}}\|_1$ can occupy out of $\|H\x\|_1$. Then $\max\limits_{K,|K|=k}\sum\limits_{i \in K} \frac{\alpha_{i}}{1+\alpha_{i}}$ is the largest fraction $\|(H\x)_{K}\|_{1}$ can occupy out of $\|H\x\|_1$ for any subset $K$ with cardinality $|K|\leq k$.  Let us then take a constant $C$ satisfying  $\frac{1}{C+1}=\max\limits_{K,|K|=k}\sum\limits_{i \in K} \frac{\alpha_{i}}{1+\alpha_{i}}$. Therefore $C\|\z_{K}\|_1 \leq \|\z_{\overline{K}}\|_1$ always holds true, for all support subsets $K$ with $|K|\leq k$, and all $\z$'s in the ranges of all $H^{local}=H_0+\Delta H$ with each row of $\Delta H$ bounded by $\epsilon$ in $\ell_2$ norm.

So now we only need to give an upper bound on the optimal objective value $\alpha_i$ of (\ref{eq:tractable_bound_fork_fundamental}). By first optimizing over $H$, $\alpha_i$ is upper bounded by the optimal objective value of the following optimization problem:
%
\begin{eqnarray}
\max_{\x}  &&\|(H_0\x)_{\{i\}}\|_{1}+ \epsilon \|\x\|_2 ,\nonumber \\
{\text{subject to}}&&\|(H_0\x)_{\overline{\{i\} }}\|_1-(n-1) \epsilon \|\x\|_2 \leq 1.
\label{eq:tractable_bound_fork}
\end{eqnarray}

Since \begin{equation*}\|(H_0\x)_{\overline{\{i\} }}\|_1 \geq \|(H_0\x)_{\overline{\{i\} }}\|_2 \geq \sigma_{min}((H_0)_{\overline{\{i\} }}) \|\x\|_2, \end{equation*} (\ref{eq:tractable_bound_fork}) is equivalent to
\begin{eqnarray}
\max_{\x}  &&\|(H_0\x)_{\{i\}}\|_{1}+ \epsilon \|\x\|_2 ,\nonumber \\
{\text{subject to}}&&\|(H_0\x)_{\overline{\{i\} }}\|_1-(n-1) \epsilon \|\x\|_2 \leq 1, \nonumber \\
&&\|(H_0)_{\overline{\{i\} }}\x\|_2-(n-1)\epsilon  \|\x\|_2 \leq 1.
\label{eq:relaxedl2norm}
\end{eqnarray}

 From $\|(H_0)_{\overline{\{i\} }}\x\|_2-(n-1)\epsilon  \|\x\|_2 \leq 1$, we know $\|\x\|_2$ in the feasible set of $(\ref{eq:relaxedl2norm})$ is upper bounded by $t'=\frac{1}{\sigma_{\min}-(n-1)\epsilon}$, where ${\sigma_{\min}}$ is the smallest singular value of the matrix  $(H_0)_{\overline{\{i\} }}$.  This enables us to further relax (\ref{eq:relaxedl2norm}) to
\begin{eqnarray}
\max_{\x}  &&\|(H_0\x)_{\{i\}}\|_{1}+ \epsilon t' ,\nonumber \\
{\text{subject to}}&&\|(H_0\x)_{\overline{\{i\} }}\|_1-(n-1) \epsilon t' \leq 1, \nonumber \\
&&(\sigma_{min}-(n-1)\epsilon)\|\x\|_2 \leq 1,
\label{eq:relaxedsemi}
\end{eqnarray}
which can be solved by semidefinite programming algorithms.
~\\

\subsubsection{Verifying the second condition in nonlinear sparse error correction}
Now we are in a position to verify the second condition for successful nonlinear sparse recovery:
\begin{equation}\label{eqn:conditiontoverify}
\frac{2(C+1)}{C-1} \frac{\sigma_{max}^{1}(H^{true}-H^{local})}{\sigma_{min}^{1}(H^{local})}\leq\beta
\end{equation}
holds for a constant $\beta< 1$ in this process.

Suppose that the Jacobian matrix at the true state $\x_0$ is $H_0$. Define a set $M_{\epsilon}$ as the set of system states $\x$ at which each row of $H^{local}-H_0$ is no larger than $\epsilon$ in $\ell_2$ norm, where  $H^{local}$ is the local Jacobian matrix.

By picking a small enough constant $\tau>0$,  we consider a diamond neighborhood $N_{\tau}=\{\x|~~\|\x-\x_0\|_1 \leq \tau \}$ of the true state $\x_0$, such that $N_{\tau} \subseteq M_{\epsilon}$. By Mean Value Theorem for multiple variables, for each $\x \in N_{\tau}$, we can find an $n \times m $ matrix $H^{true}$ such that $h(\x)-h(\x_0)=H^{true}(\x-\x_0)$, and moreover, each row of $H^{true}-H_0$ is also upper bounded by $\epsilon$ in $\ell_2$ norm.

For this neighborhood $N_{\tau}$, using the proposed semidefinite programming (\ref{eq:relaxedsemi}), for a  fixed $k$ (the sparsity), we can find a $C$ such that for every local Jacobian $H^{local}$ in that neighborhood, every vector $\z$ in the subspace generated by $H^{local}$ satisfies the balancedness property with the parameter $C$.

Meanwhile, $\sigma_{max}^{1}(H^{true}-H^{local}) \leq n\times 2\epsilon$ in that neighborhood since each row of $H^{true}-H_0$ and $H^{local}-H_0$ is no larger than $\epsilon$ in $\ell_2$ norm.

Similarly, for every $H^{local}$ in that neighborhood $N_{\tau}$, by the definition of singular values and inequalities between $\ell_1$ and $\ell_2$ norms,  $\sigma_{min}^{1}(H^{local}) \geq \sigma_{min}^{1}(H_0)-n\epsilon \geq \frac{1}{\sqrt{n}} \sigma_{min}(H_0)-n\epsilon$, where $\sigma_{min}(H_0)$ is the smallest singular value of $H_0$.

So as long as
\begin{equation}\label{eqn:conditiontoverifyfinal}
\frac{2(C+1)}{C-1} \frac{  2 n \epsilon  }{\frac{1}{\sqrt{n}}\sigma_{min}(H_0)-n\epsilon}\leq\beta<1,
\end{equation}
once our state estimation algorithm gets into $N_{\tau}=\{\x| ~~\|\x-\x_0\|_1 \leq \tau \} \subseteq M_{\epsilon}$, the algorithm will always stay in the region $N_{\tau}$ (because $\|\x^k-\x_0\|_1$ decreases after each iteration, see the proof of Theorem \ref{thm:l1iterative} with $\epsilon=0$), and the iterative state estimation algorithm will converge to the true state $\x_0$.


In summary, we can explicitly compute a parameter $\epsilon$ such that inside a diamond neighborhood $N_{\tau}=\{\x| \|\x-\x_0\|_1 \leq \tau \} \subseteq M_{\epsilon}$, the condition $\frac{2(C+1)}{C-1} \frac{\sigma_{max}^{1}(H^{true}-H^{local})}{\sigma_{min}^{1}(H^{local})}\leq\beta <1$ is always satisfied and the decoding algorithm always converges to the true state $\x$ once it gets into the region  $N_{\tau}$. The size of the region $N_{\tau}$ depends on specific functions. For example, if $h(\cdot)$ is a linear function and $C>1$, $N_{\tau}$ can be taken as $R^{m}$. This fits with known results in \cite{CandesErrorCorrection,CT1} for linear measurements where the local Jacobian is the same everywhere.

\section{Numerical Results}
\label{sec:numerical}

In our simulation, we apply (\ref{eq:bus}) to estimate an
unknown vector from Gaussian linear measurements with both  sparse
errors and noise, and also apply the iterative method to recover
state information from nonlinear measurements with bad data and
noise in a power system.

\textbf{Linear System:} We first consider recovering a signal vector
from linear Gaussian measurements. Let $m=60$ and $n=150$. We
generate the measurement matrix $H^{n \times m}$ with i.i.d.
${N}(0,1)$ entries. We also generate a vector $\bfx\in {R}^m$ with
i.i.d  entries uniformly chosen from interval $[-1,1]$. 


 We add to each
measurement of $H\bfx$ with a Gaussian noise independently drawn
from ${N}(0, \sigma^2)$. Let $\rho$ denote the percentage of erroneous
measurements. Given $\rho$, we randomly choose $\rho n$
measurements, and each such measurement is added with a Gaussian
error independently drawn from ${N}(0, 4^2)$. We apply (\ref{eq:bus})
to estimate $\bfx$. Since the noise vector $\bfz \in {R}^m$ has i.i.d. ${N}(0, \sigma^2)$ entries, then $\|\bfz\|_2/\sigma$ follows the chi distribution with dimension $m$. Let $F_m(x)$ denote the  cumulative distribution function (CDF) of the chi distribution of dimension $m$, and let $F_m^{-1}(y), y\in [0,1]$ denote its inverse distribution function. We choose $\epsilon$ to be $C\sigma$ where
\begin{equation}\label{eqn:C}
C:=F_m^{-1}(0.98).
\end{equation}
Thus, $\|\bfz\|_2 \leq \epsilon$ holds with probability 0.98 for randomly generated $\bfz$. Let $\bfx^*$ denote the estimation of $\bfx$, and the relative
estimation error is represented by $\|\bfx^*-\bfx||_2/\|\bfx\|_2$. 

We  fix the noise level and consider the estimation performance when
the number of erroneous measurements changes. Fig. \ref{fig:gaussianconstraint} shows how the estimation error
changes as $\rho$ increases, and each result is averaged over one hundred runs. We choose different $\sigma$ between 0 and 2. When $\sigma=0$, the measurements have errors but no noise, and (\ref{eq:bus}) is reduced to conventional $\ell_1$-minimization problem. One can see from Fig. \ref{fig:gaussianconstraint} that when there is no noise, $\bfx$ can be correctly recovered from (\ref{eq:bus}) even when twenty  percent of measurements contain errors.

\begin{figure}[t]
\centering
\includegraphics[scale=0.4]{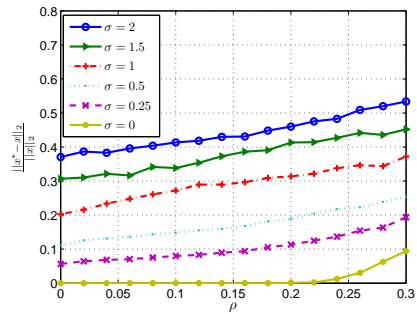}
\caption{Estimation error versus $\rho$ for Gaussian measurements with fixed noise level}\label{fig:gaussianconstraint} 
\end{figure}

We next compare the recovery performance of (\ref{eq:bus})  with that of $\ell_1$-minimization. We fix the number of erroneous measurements to be two, i.e., the percentage of erroneous
measurements is $\rho=0.0133$, and increase the noise level $\sigma$. One can see from  Fig. \ref{fig:comparel1} that (\ref{eq:bus}) has a smaller estimation error than that of $\ell_1$-minimization when the noise level is not too small. This is not that surprising since (\ref{eq:bus}) takes into account the measurement noise while $\ell_1$-minimization does not.
\begin{figure}[t]
\centering
\includegraphics[scale=0.4]{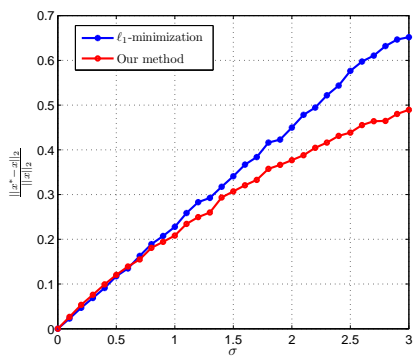}
\caption{Comparison of the estimation error with $\ell_1$-minimization}\label{fig:comparel1}
\end{figure}

\textbf{Power System:} We also consider state estimation in power networks. 
Monitoring the system characteristics  is a fundamental prerequisite for the efficient and reliable operation of the power networks. Based on the meter measurements, state estimation provides pertinent information, e.g., the voltage magnitudes and
the voltage angles at each bus, about the operation  condition of a power grid \cite{AE04,M00,SH74}. Modern devices like a phasor measurement unit (PMU) \cite{JYLLM00} can measure the system states directly, but their installation in the current power system is very limited. 
The current meter measurements are mainly bus voltage magnitudes, real and
reactive power injections at each bus, and the real and reactive
power flows on the lines. The meter measurements are subject to observation noise, and may contain large errors \cite{BC,HSKF75,MG}. Moreover, with the modern information technology introduced in a smart grid, there may exist data attacks from intruders or malicious insiders \cite{TL10a,TL10b,LNR11}. Incorrect output of state estimation will   mislead the system operator, and can possibly result in catastrophic consequences.  Thus, it is important to obtain accurate estimates of the system states from measurements that are noisy and erroneous. If the percentage of the erroneous measurements is small, which is a valid assumption given the massive scale of the power networks, the state estimation is indeed a special case of the sparse error correction problem from nonlinear measurements that we discussed in Section \ref{sec:nonlinear}. 

The relationship between the measurements and the state variables
for a $k'$-bus system can be stated as follows \cite{KV82}: {\small
\begin{eqnarray}\label{eqn:sys1}
P_i&=&\sum_{j=1}^{k'} E_i E_j Y_{ij} \cos(\theta_{ij} +\delta_i
-\delta_j),\\
Q_i&=&\sum_{j=1}^{k'} E_i E_j Y_{ij} \sin(\theta_{ij} +\delta_i
-\delta_j),\\
P_{ij}&=&E_i E_j Y_{ij} \cos(\theta_{ij} +\delta_i -\delta_j) \nonumber\\
&&-E_i^2 Y_{ij}\cos \theta_{ij} +E_i^2 Y_{si} \cos\theta_{si} \quad
i \neq j,\\
Q_{ij}&=&E_i E_j Y_{ij} \sin(\theta_{ij} +\delta_i -\delta_j)
\nonumber \\&&-E_i^2 Y_{ij}\sin \theta_{ij} +E_i^2 Y_{si}
\sin\theta_{si} \quad i \neq j, \label{eqn:sys2}
\end{eqnarray}
} where $P_i$ and $Q_i$ are the real and reactive power injection at
bus $i$ respectively, $P_{ij}$ and $Q_{ij}$ are the real and
reactive power flow from bus $i$ to bus $j$, $E_i$ and $\delta_i$
are the voltage magnitude and angle at bus $i$. $Y_{ij}$ and
$\theta_{ij}$ are the magnitude and phase angle of admittance from
bus $i$ to bus $j$, $Y_{si}$ and $\theta_{si}$ are the magnitude and
angle of the shunt admittance of line at bus $i$. Given a power
system, all $Y_{ij}$, $\theta_{ij}$, $Y_{si}$ and $\theta_{si}$ are
known.

For a $k'$-bus system, we treat one bus as the reference bus and set
the voltage angle at the reference
bus to be zero. 
There are $m=2k'-1$ state variables with the first $k'$ variables
for the bus voltage magnitudes $E_i$ and the
rest $k'-1$ variables for the bus voltage angles $\theta_i$. 
Let $\bfx \in {R}^{m}$ denote the state variables and let $\bfy \in
{R}^{n}$ denote the  $n$ measurements of the real and reactive power
injection and power flow.
 Let $\bfv \in {R}^n$ denote the noise
and $\bfe \in {R}^n$ denote the sparse error vector. Then we can
write the equations in a compact form,
\begin{equation}
\bfy=h(\bfx)+\bfv+\bfe,
\end{equation}
where $h(\cdot)$ denotes $n$ nonlinear functions defined in
(\ref{eqn:sys1}) to (\ref{eqn:sys2}). We apply the iterative $\ell_1$-minimization algorithm in Section \ref{subsec:algorithmdescription} to recover $\bfx$ from $\bfy$.

\begin{figure*}[t]
\centering
\includegraphics[scale=0.4]{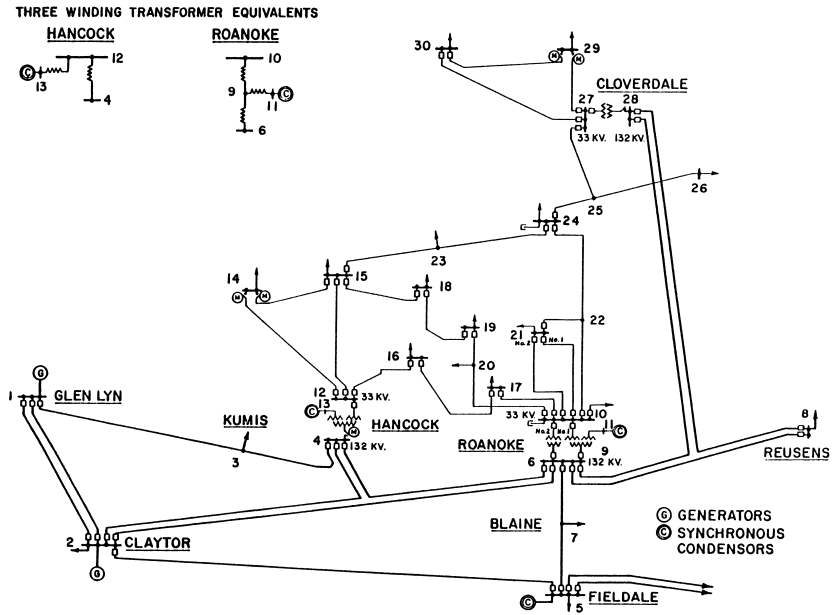}
\caption{IEEE 30-bus test system}\label{fig:bus}
\end{figure*}

We evaluate the performance on the IEEE 30-bus test system. Fig.
\ref{fig:bus} shows the structure of the test system. Then the state
vector $\bfx$ contains $59$ variables. Its first thirty entries correspond to the normalized bus voltage magnitudes, and are all close to 1. Among the thirty buses in this example, the minimum bus voltage magnitude is 0.992, and the maximum bus voltage magnitude is 1.082. The last twenty-nine entries of $\bfx$ correspond to the relative bus voltage angles to a  reference bus. Here, the bus voltage angles are in the range -0.0956 to -0.03131.  We take $n=100$ measurements
including the real and reactive power injection at each
bus and some of the real and reactive power flows on the lines. We first characterize the dependence of the estimation performance on the noise level when the number of erroneous measurements is fixed. For each fixed $\rho$,
 we randomly choose a
set $T$ 
with cardinality $|T|=\rho n$. Each measurement with its index in
$T$ contains a Gaussian error independently drawn from
${N}(0,0.5^2)$.
Each measurement also contains a Gaussian noise independently drawn
from ${N}(0,\sigma^2)$. For a given noise level $\sigma$, we apply the iterative $\ell_1$-minimization algorithm in Section \ref{subsec:algorithmdescription}
to recover the state vector $\bfx$ with $\epsilon=C\sigma$ where $C$ is defined in (\ref{eqn:C}). The initial estimate $\bfx^0$ is chosen to have `1's in its first thirty entries and `0's in its last twenty-nine entries. The relative error of the initial estimate is $\|\bfx^0-\bfx\|_2/\|\bfx\|_2=0.2447$.  
For example, in  one realization when $\rho=0.02$ and $\sigma=0$,  the iterative method takes seven iterations to converge. Let $\bfx^k$ be the estimate of $\bfx$ after the $k$th iteration, $\bfx^7$ is treated as the final estimate $\bfx^*$. The relative estimation errors $\|\bfx^k-\bfx\|_2/\|\bfx\|_2$ are  $0.1208$,   $0.0179$,    $0.0049$,   $0.0017$,   $0.00017$, $3 \times 10^{-6}$, and  $3\times 10^{-8}$ respectively.
 Fig.
\ref{fig:lambda2sigma} shows the relative estimation error $\|\bfx^*-\bfx\|_2/\|\bfx\|_2$ against $\sigma$ for various $\rho$. The result is averaged over two hundred runs. Note that when $\sigma=0$, i.e., the measurements contain sparse errors but no observation noise, the relative estimation error is not zero. For example,  $\|\bfx^*-\bfx\|_2/\|\bfx\|_2$ is 0.018 when $\rho$ is 0.02. That is because the system in Fig.
\ref{fig:bus} is not proven to satisfy the condition in Theorem \ref{thm:l1iterative}, and the exact recovery with sparse bad measurements is not guaranteed here. However, our next simulation result indicates that our method indeed outperforms some existing method in some cases.  
\begin{figure}[t]
\centering
\includegraphics[scale=0.4]{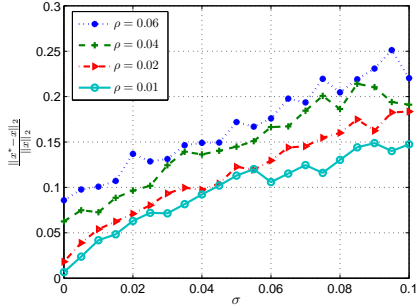}
\caption{Estimation error versus $\sigma$ with fixed percentage of
errors in power system}\label{fig:lambda2sigma}
\end{figure}

We compare our proposed method with two other recovery methods. In the first method, we assume we know the location of the errors,  i.e., the support $T$ of the sparse vector $\bfe$. We delete these erroneous measurements and apply the conventional Weighted Least Squares (WLS) method (Section III-A of \cite{BC}) to estimate system states based on the remaining measurements. The solution minimizing the weighted least squares function requires the application of an iterative method, and we follow  the updating rule in \cite{BC} (equations (6) and (7)). In the second method, we apply the $\hat{b}$ test in \cite{MG} to detect the erroneous measurements. It first applies the WLS method to estimate the system states, then applies a statistical test to each measurement. If some measurement does not pass the test, it is identified as a potential erroneous measurement, and the system states are recomputed by WLS based on the remaining measurements. This procedure is repeated until every remaining measurement passes the statistical test. Figure \ref{fig:compare} shows the recovery performance of three different methods when $\rho=0.02$. The result is averaged over two-hundred runs. The estimation error of our method is less than that of  the $\hat{b}$ test. WLS method with known error location has the best recovery performance. 

\begin{figure}[t]
\centering
\includegraphics[scale=0.4]{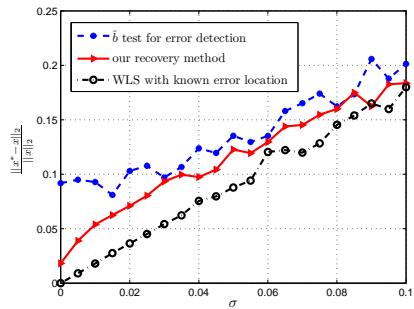}
\caption{Estimation error of three recovery methods}\label{fig:compare}
\end{figure}

\begin{figure}
\centering
\includegraphics[scale=0.4]{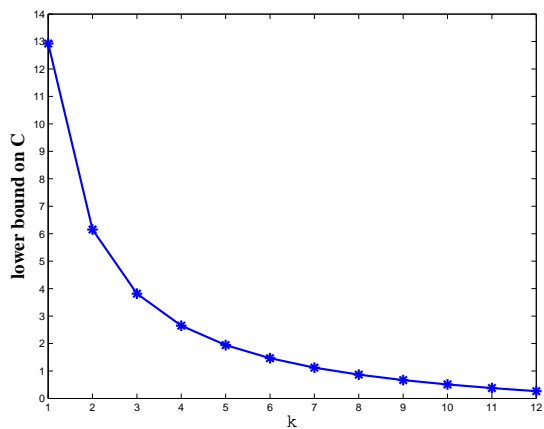}
\caption{Lower bound on $C$ when $n=200$, $m=50$ and $\epsilon=10^{-3}$}\label{fig:kCrelation}
\end{figure}

\textbf{Verification of the Conditions for Nonlinear Sparse Error Correction:} We use the verification methods in Subsection \ref{subsec:verification} to verify the conditions for nonlinear sparse error correction. The Jacobian matrix $H_0$ for the true state is generated a $200 \times 50$ matrix with i.i.d. standard Gaussian entries. We also assume that each row of $\Delta H=H^{local}-H_0$ is upper bounded by $\epsilon=10^{-3}$ in $\ell_2$ norm. We then plot the sparsity $k$ and the corresponding lower bound on $C(k)$ in Figure \ref{fig:kCrelation} such that for every $H^{local}$, every index set $K$ with cardinality $|K|=k$, and every vector $\z$ in the subspace spanned by $H^{local}$, we always have $C\|\z_{K}\|_1 \leq \|\z_{\overline{K}}\|_1$. We can see that for $k=7$, $C$ is lower bounded by $1.1222$. So when there are no more than $7$ bad data entries, there exists a convergence region around the true state such that the iterative algorithm converges to the true state once the iterative algorithm gets into that convergence region. Of course, we remark that the computation here is pretty conservative and in practice, we can get much broader convergence scenarios than this theory predicts. However, this is a meaningful step towards establishing sparse recovery conditions for nonlinear measurements.



\section*{Acknowledgment}
The research is supported by NSF under CCF-0835706, ONR under N00014-11-1-0131, and DTRA under HDTRA1-11-1-0030.

\bibliographystyle{IEEEbib}

\end{document}